\documentclass[journal]{IEEEtran}
\usepackage{amsmath}    
\usepackage{graphics,epsfig,subfigure}    
\usepackage{verbatim}   
\usepackage{color}      

\usepackage{subfigure}  

\usepackage{amssymb}
\usepackage{epstopdf}
\usepackage{graphicx}
\graphicspath{{/Figures/}}
\usepackage{wasysym}
\usepackage{wrapfig}
\usepackage{sidecap}
\usepackage{float}
\usepackage{subfigure}
\usepackage{enumerate}
\usepackage{graphicx}
\graphicspath{{./figures/}}
\usepackage{algorithmic}
\usepackage{algorithm}
\usepackage{amsthm}
\usepackage{bookmark}

\newtheorem{lemma}{\textbf{Lemma}}
\newtheorem{theorem}{\textbf{Theorem}}
\newtheorem{corollary}{\textbf{Corollary}}

\begin{document}
\title{Capacity Estimation for Vehicle-to-Grid Frequency Regulation Services with Smart Charging Mechanism}

\author{Albert Y.S. Lam, Ka-Cheong Leung, and Victor O.K. Li
\thanks{A preliminary version of this paper was presented in \cite{smartgridcomm2012}.}
\thanks{A.Y.S. Lam is with the Department of Computer Science, Hong Kong Baptist University, Kowloon Tong, Hong Kong (e-mail: albertlam@ieee.org).}
\thanks{K.-C. Leung and V.O.K. Li are with the Department
of Electric and Electronic Engineering, The University of Hong Kong, Pokfulam Road, 
Hong Kong (e-mail: \{kcleung, vli\}@eee.hku.hk).}
}


\maketitle

\begin{abstract}
Due to various green initiatives, renewable energy will be massively incorporated into the future smart grid.  However, the intermittency of the renewables may result in power imbalance, thus adversely affecting the stability of a power system.  Frequency regulation may be used to maintain the power balance at all times.  As electric vehicles (EVs) become popular, they may be connected to the grid to form a vehicle-to-grid (V2G) system.  An aggregation of EVs can be coordinated to provide frequency regulation services.  However, V2G is a dynamic system where the participating EVs come and go independently. Thus it is not easy to estimate the regulation capacities for V2G. In a preliminary study, we modeled an aggregation of EVs with a queueing network, whose structure allows us to estimate the capacities for regulation-up and regulation-down, separately.  The estimated capacities from the V2G system can be used for establishing a regulation contract between an aggregator and the grid operator, and facilitating a new business model for V2G. In this paper, we extend our previous development by designing a smart charging mechanism which can adapt to given characteristics of the EVs and make the performance of the actual system follow the analytical model. 
\end{abstract}

\IEEEpeerreviewmaketitle

\section{Introduction}

\IEEEPARstart{F}{or}  
a reliable power system, power balancing needs to be maintained at all
times; power generation and consumption must always be equal.  
Traditional power generations (e.g., thermal power stations) and the renewables serve in
the day-ahead market~\cite{CA2011energy}.  One of the most
challenging problems of incorporating the renewables into the power system is
their intermittency, rendering it difficult to predict the amount of power
generated from the renewables accurately.  It is possible that the resulting
generation from the day-ahead market are excessive or deficient
compared with the predicted amount.  The real-time market bridges the residual gap
between the power generation and the actual demand, accomplished by
the ancillary services, including frequency regulation, spinning reserve,
supplemental reserve, replacement reserve, and voltage
control~\cite{frequencyregulation}.  According to the United States (U.S.)
Federal Energy
Regulatory Commission, ancillary services are ``those services necessary to
support the transmission of electric power from seller to purchaser given the
obligations of control areas and transmitting utilities within those control
areas to maintain reliable operations of the interconnected transmission
system''~\cite{FERC}.  Regarding load balancing, spinning, supplemental
reserve, and replacement reserves are for contingency purposes while frequency
regulation tracks on a minute-to-minute basis.  In this paper, we
focus on the ancillary service given by frequency regulation.

For a power system to function properly, the operating frequency should be maintained close to its nominal value, e.g., 60 Hz in the U.S. excessive power generated (i.e., generation is larger than consumption) will drive the system frequency higher than the nominal setting while a deficiency of generation results in a smaller system frequency. Frequency regulation is the measure of adjusting the system frequency to the nominal value by providing small power (positive or negative) injections into the grid. Many Regional Transmission Organizations (RTOs) and Independent System Operators (ISOs), e.g., PJM, simply call this service ``regulation'' \cite{PJM}.
The balance of generation and demand between control areas is measured in terms of area control error (ACE). Each control area generates automatic generation control (AGC) signals based on its ACE values and the regulation resources respond to the AGC signals to perform regulation. This is achieved through a real-time telemetry system and controlled by the grid operator.
Purchase and sale of regulation services are accomplished in the regulation market managed by the corresponding ISO/RTO. Consider PJM as an example \cite{PJM}. Resource owners submit offers to the Market Clearing Engine, which optimizes the RTO dispatch profile and determines corresponding clearing prices. The market is cleared between the regulation resources and service purchasers with the clearing prices. The details can be found in \cite{PJM,AGC,iagc}.

There have been some studies about integrating renewables into the grid more
reliably and efficiently, such as \cite{sellingwind}.  One proposed solution
is the introduction of energy storage to \emph{defer} the excess for the
future deficient.  Examples of energy storage include batteries, flywheels, and
pumped storage.  In the near future, one of the most realistic forms is
batteries.  This can be justified by the expanding markets of plug-in hybrid
electric vehicles or simply electric vehicles (EVs).  For example, it is
forecast that there will be 2.7 million EVs on the road in the U.S. by
2020~\cite{evnum}.  In California, it is expected that roughly 70\% of new
light-duty vehicles and 60\% of the fleets will be EVs~\cite{CA2011energy}.  
The integration of EVs into the power grid is called the vehicle-to-grid (V2G)
system depicted in Fig.~1 of \cite{smartgridcomm2012}.

Frequency regulation requires power in the order of MW while each EV can only supply
power around 10-20~kW~\cite{battery}.  In order to provide regulation service
from the V2G system, an aggregation of EVs is necessary and an aggregator
coordinates a group of EVs.  The aggregators thus provide regulation services
to the grid, which are controlled and coordinated by the operators.  In
general, an aggregator can be a parking structure or a facility
coordinating the EV activities of the households in a residential area.

To implement regulation in V2G, the aggregators need to make contracts with the
grid operators.  The V2G system can support both \emph{regulation-up} (RU) and
\emph{regulation-down} (RD) services.  The former means that the grid does not have
enough power supply and extra power sources (e.g., V2G) provide the shortfall.
The latter refers to the situation in which extra power loads are needed to
absorb the excessive power.

According to \cite{MISObook}, short-term stored-energy resources, e.g., batteries, are excellent candidates of regulation resources due to their very fast response time to AGC and their capability of reducing $\text{CO}_2$ emissions. Many ISOs, including New York ISO, ISO New England, and California ISO, have been integrating short-term energy storage resources into their regulation markets.
A team, composed of experts from the University of Delaware and PJM, conducted a practical demonstration to show that V2G is capable of providing real-time frequency regulation \cite{delaware}.
Experienced distribution engineers pointed out that an EV has no difference from a distributed generator or additional load for regulation in the technological viewpoint \cite{delaware}.
Moreover, \cite{financialreturn} revealed that V2G has significant potential for financial return from frequency regulation.
Based on \cite{traveltrends}, private vehicles in the U.S. are driven less than an hour a day on the average and this implies that we may intelligently utilize EVs for other purposes when they are parked idly. 
From all these, we can see that V2G, as a regulation resource, will not bring significant technological challenges to the existing power system when regulation is taking place. Moreover, it has a great financial incentive to be implemented in the practical power system.

In the regulation market, a regulation resource owner needs to submit the capacity of its resource for the locational marginal price forecasts, which are in turn used to determine the market clearing prices.
One type of charges for the regulation services is capacity
payment~\cite{V2G_fundamentals}.  It refers to the service
charges due to the V2G system only guaranteeing power support when the grid
requires RU or RD.  In other words, the V2G system gets paid even
without any actual power transfer.  The grid operator pays for the service
according to the expected amount of power to be supplied and absorbed for
RU and RD, respectively.  
Unlike the traditional generators with high controllability, it is not easy to determine the regulation capacities for V2G. The main reason is that V2G is a dynamic system, in which the governed EVs are all autonomous. The number of EVs managed by an aggregator varies from time to time and thus actual regulation capacities contributed by the EVs are also varying. Although V2G is practically feasible \cite{delaware} and financially favorable \cite{financialreturn}, we need to estimate the regulation capacity of V2G so as to put it into operation extensively.
In this paper, we focus on capacity estimation of the V2G system for an aggregator which can help estimate the
total profit and set up a contract between an aggregator and the grid
operator. 
A smart charging mechanism is designed to enhance the flexibility of the system.
 Due to the dynamics of EVs and the similarities between the
batteries of V2G (for power) and the buffers of the communication networks
(for data packets), we estimate the V2G capacity for regulation with the
queueing theoretic approach, which has been widely used for performance
analysis in communication networks~\cite{networkqueueing}.

The rest of this paper is organized as follows.  We give some related work of
 V2G studies in Section~\ref{sec:related} and a system overview in
Section~\ref{sec:overview}.  Section~\ref{sec:model} presents our analytical
model with the RU and RD capacities derived. The smart charging mechanism is discussed in Section~\ref{sec:scm} and
a performance study of the V2G system for the power regulation
services is presented in Section~\ref{sec:perf}.  Finally,
Section~\ref{sec:concl} concludes our work.

\section{Related Work} \label{sec:related}

The preliminary version of this work can be found in \cite{smartgridcomm2012}.
In \cite{smartgridcomm2012}, we defined a queueing network model to estimate the RU and RD capacities. However, we assumed that there exists a smart charging mechanism which makes the service times at various queues exponentially distributed. This exponential distribution property is one of the keys to developing mathematically tractable closed-form solutions for the capacities. In this paper, we relax this assumption by explaining how such  smart charging mechanism works. It allows the model to function even when the attributes of EVs are distributed in unknown distributions. We also perform simulation to verify the behavior of this mechanism when applied to various queues in the model.

There are many studies on V2G since it is expected to be a major component in
the future smart grid.  In \cite{V2G_fundamentals} and
\cite{V2G_implementation}, V2G was systematically introduced with studies on the business model for V2G.  They gave information of
different kinds of EVs and different power markets, including baseload power,
peak power, spinning reserves, and regulation.  The merits of V2G are quick
response and high-value services with low capital costs, but V2G has shorter
lifespans and higher operating costs per kWh.  
They gave some rough idea about the scale of V2G so as to make it comparable
with the traditional regulation from generators.  
V2G energy trading was studied as an auction in \cite{doublelayer}. 
Interested readers can refer
to \cite{impact_review} for a comprehensive review on the impact of V2G on
distribution systems and utility interfaces of power systems.

Queueing theory has been used to study the aggregate behavior of EVs.
In \cite{EV_loadflow}, a simple $M/M/c$ queueing model for EV charging was
devised and a similar idea was also adopted in \cite{V2GParking} to determine V2G capacity. 
Ref. \cite{chargingmodel} suggested an $M/M/\infty$ queue with random interruptions to model the EV charging process and analyzed the dynamics with time-scale decomposition. 
The ``service'' process adopted was assumed to be exponential but it may not be practical unless there is a special arrangement to conform with the exponential property. In this work, we will design a smart charging mechanism to overcome this problem.

Some existing work investigated V2G frequency regulation and the capacities of V2G systems.
In \cite{optimal_aggregator}, an optimal charging control scheme for maximizing the revenue of an EV from supporting frequency regulation was proposed.  
In \cite{V2G_QP}, the problem was formulated as a quadratic program and an efficient algorithm considering discharge was devised.
Ref. \cite{estimation_capacity} considered the user pattern to develop an approximate probabilistic model for achievable power capacity, which was then utilized to determine the contracted power capacity.
However, this contracted power capacity (similar to the capacity aimed to be achieved in this work) depends on the payment methodology, which is a business consideration rather than an engineering issue.  
In \cite{V2G_fundamentals}, the V2G capacity was achieved by taking the maximum of the current-carrying capacity of the connecting wires and the power available from the vehicle battery.  The former can be improved with advances in wiring technology. The latter simply converts the stored energy into power with the consideration of efficiency factors.
However, the achieved capacity in \cite{V2G_fundamentals} is for one single vehicle, which is assumed to be static. For an aggregation of EVs with stochastic arrivals and departures, this estimation may not be very practical.
In \cite{V2GCap_dynamicEV}, the V2G capacity was estimated using dynamic EV scheduling, where the design was mainly targeted on the building energy management system. Its proposed algorithm utilized  the forecasted building load demand and EV charging profiles to estimate the V2G capacity. This design may not be applicable to frequency regulation as load demand of the supported region may not be available to an aggregator.
In this work, we  compute the separate RU and RD capacities for an aggregation of EVs without any assumptions on the stored energy of EVs and their durations of stay. Moreover, our computation does not  rely on any payment methodology and this provides flexibility for other business considerations.
To the best of our knowledge, there is no unified study on the capacity management for both RU and RD in the V2G system.

There are many studies about coordinated EV charging.
The issue of time-varying electricity price signals on cost-optimal charging was considered in~\cite{time-varying_price}. 
A decentralized algorithm for coordinating the charging/discharging schedules of EVs in order to meet the regulation demand was devised in~\cite{online_scheduling}. 
Refs. \cite{charging_net} and \cite{charging_flow} focused on the network aspects.
Ref. \cite{charging_net} investigated the impact of EV charging on distribution grids and showed that controlled charging can result in significant reduction of overloaded network components. 
In \cite{charging_flow}, the authors proposed a linear-approximation-based framework for online adaptive EV charging. It could reduce the violations on various network limits, e.g., flow limit and voltage magnitude limit, due to high penetration of EVs.
Ref. \cite{charging_price} performed EV charging according to the electricity price. Cao et al. \cite{charging_price} proposed a control method to coordinate charging with the time-of-use market price.
Ref. \cite{charging_building} studied a EV charging method for smart homes and buildings in the presence of photovoltaic systems.
In \cite{charging_mobility}, Miao et al. proposed a mobility-aware charging strategy for globally optimal energy utilization by means of appropriately routing mobile EVs to charging stations with the assistance of vehicular ad-hoc network.
All these efforts attempt to design EV charging strategies based on the various system objectives. In this paper, we take a different perspective; we construct a smart charging mechanism for capacity management for V2G regulation services.

\section{System Overview} \label{sec:overview}

Each EV is assumed to be autonomous.  It can participate and leave the V2G
system according to the schedule of the EV owner.  Once an EV is connected or
plugged to the system, it will be actively charged and/or support regulation
until it departs.  When being actively charged, it pays for the amount of energy
consumed.  While it is supporting the regulation, it receives payment for
providing the service.  Regulation can result in either EV charging or
discharging, depending on whether RU or RD is requested.  
Thus, it is possible for an EV to get paid while it is being charged (i.e., in an
RD event). Charging events in which an EV requests
charging itself and supports regulation are called \textit{active charging} and
\textit{reactive charging}, respectively.  
A discharging event happens only
when an EV participates in supporting RU.  Both the residual energy
stored in the battery of the EV at the time it arrives and the energy charged
from active or reactive charging can be used to support RU.  However,
an EV cannot support RU when its battery is fully discharged.
Similarly, an EV cannot participate in supporting RD when its
battery is fully charged up.  Since the battery capacity is finite, the amount
of energy stored in a battery affects its potential for supporting the
RU and RD services.  In this paper, we aim to estimate
the capacities of an aggregator for regulation services so that it will be
beneficial for an aggregator to establish a contract with the grid operators.
Hence, we only consider the events of active charging.  The charging and
discharging rates due to regulation are small enough so that the estimated
capacities for regulation services are not affected by charging and discharging
events due to regulation.

Now we focus on a particular aggregator.  We denote the set of EVs, each of
which has registered at the aggregator for providing the regulation service, by
$\mathcal{I}$.  The events associated with each EV and among EVs are
independent with each other.  
All EVs are assumed to be heterogeneous such that
they can be equipped with batteries of different capacities when fully charged.  
The state-of-charge (SOC) of an EV refers to the amount of energy stored in its
battery normalized with the maximum capacity.  We denote SOC of EV~$i$ at
time~$t$ by $x_i(t)$.  Without loss of generality, we assume
$x_i(t) \in [0,1], \forall i\in \mathcal{I}$.  We also define the
\textit{target SOC} of EV~$i$ as the amount of energy, normalized with the
maximum battery capacity, that the EV's owner aims to reach when it departs,
given in a range $[\underline{x}_i, \overline{x}_i]$, where $\underline{x}_i$
and $\overline{x}_i$ are the lower and upper limits of the target SOC of
EV~$i$, and $0\leq \underline{x}_i \leq \overline{x}_i \leq 1$.  In other
words, if EV~$i$ leaves the system at time~$t'$, it aims to satisfy
$\underline{x}_i \leq x_i(t') \leq \overline{x}_i$.  If the target SOC is
merely a value, we have $\underline{x}_i = \overline{x}_i$. 
 
The lower SOC target threshold $\underline{x}_i$ represents the minimum
targeted amount of energy, normalized with the maximum battery capacity,
retained for EV~$i$ when it departs from the system.  Hence, it is designed to
meet the mobility pattern of EV~$i$.  For example, an EV which travels a lot in
between two successive chargings requires a higher $\underline{x}_i$.  If an EV
can be charged quite frequently, a lower $\underline{x}_i$ may be sufficient to
support its operation.  On the other hand, $\overline{x}_i$ is defined for
regulation.  Recall that a fully charged EV cannot provide the RD
service.  $\overline{x}_i<1$ means that EV~$i$ reserves room of size
$(1 - \overline{x}_i)$ for later RD opportunities or other
purposes.  At time~$t$, if $x_i(t)$ is smaller than $\underline{x}_i$, active
charging always happens in order to bring SOC to the target range.  However,
active charging must stop when $x_i(t)$ reaches $\overline{x}_i$, since no
future RU event is guaranteed to happen in order to bring SOC back
to the target range.
 
Recall that regulation can result in charging or discharging to an EV.  We can
increase $x_i(t)$ by both active and reactive (i.e., RD)
chargings, while we can only reduce $x_i(t)$ by RU.  Hence, when
EV~$i$ is still connected to the system, it will be in one of three states
according to the value of $x_i(t)$, each of which supports different 
regulation services, as follows:
\begin{itemize}[]
\item State 1) $x_i(t) \leq \underline{x}_i$:
Only RD (reactive charging) is allowed.

\item State 2) $\underline{x}_i< x_i(t) < \overline{x}_i$:
Both RU and RD are allowed.

\item State 3) $x_i(t)\geq \overline{x}_i$:
Only RU (discharging) is allowed. 
\end{itemize} 

For simplicity in the analysis, we do not consider that the EVs are actually
charged or discharged due to regulation in this paper.  
Let $r_i(t) \geq 0$ be
the active normalized charging rate of EV~$i$ at time~$t$ and it is constant
over time in each state, i.e., $r_i(t) = r_i, t \geq 0$.
Consider that EV~$i$ is plugged in
at time~$t$ and its SOC is $x_i(t)$.  If it is actively charged at rate $r_i$,
after a time period $\Delta t$, we have
$x_i(t+\Delta t) = x_i(t) + r_i \Delta t$.

From the standpoint of an EV owner, the primary concern is to charge its EV
such that it has enough battery level to support its operation.  The profit
derived from providing the ancillary services is of secondary concern.  In
other words, an owner considers to provide the ancillary services from its EV
only if the remaining energy (after discharging from providing the ancillary
services) is enough to support its operation.  Hence, we propose the following
simple charging policy: When EV~$i$ arrives at the system with SOC below
$\overline{x}_i$, it will be actively charged until $\overline{x}_i$ is
reached.  Otherwise, no active charging is required.

In fact, we can always set $\underline{x}_i = \overline{x}_i$ to simplify the
system.  EV~$i$ supports RD when $x_i(t)$ is below
$\overline{x}_i$, and it supports RU when $x_i(t)$ goes above
$\overline{x}_i$.  However, suppose $x_i(t) = \overline{x}_i$
when the system is supporting regulation (i.e., it can be actually charged or discharged due to
regulation).  When there exists a random sequence of RU and
RD requests,
the EV will be oscillating between States~1 and 3 previously discussed and
this will make the system unstable.  The introduction of State 2 can help stabilize
the system.

Note that we aim to perform capacity management by estimating the capacities
for regulation to help construct a contract between an aggregator and a grid
operator.  There are different kinds of regulation contracts in the market: 
\begin{itemize}
\item RD only: An EV always absorbs power from the grid to provide the service.  To maximize
the profit, we can simply set $\underline{x}_i = \overline{x}_i = 0$ so as to
reserve the largest room for energy absorption.

\item RU only: An EV always supplies power to the grid when providing the service.  To
maximize the profit, we can simply set $\underline{x}_i = \overline{x}_i = 1$
to preserve as much energy in the battery as possible for future discharging
events.

\item RU and RD: Both RU and RD are allowed.  We would set
$0 < \overline{x}_i <1$ appropriately to balance the demand for RU
and RD. 
\end{itemize}

We consider the V2G system supporting both RU and
RD.  We can define two kinds of capacities for the V2G regulation
services, namely, the \textit{RD capacity} and
\textit{RU capacity}.  The former refers to the total amount of
energy that can be absorbed by the system to support RD.
Similarly, the latter corresponds to the total amount of energy available from
the system to support RU.  Here, we focus on determining the
RU and RD capacities of one particular aggregator.  The
capacity of the whole V2G system can then be seen as the sum of the capacities
of the individual aggregators.  In the next section, we propose an analytical
model to estimate the two capacities of an aggregator for our charging policy.

\section{Analytical Model} \label{sec:model}

In this section, we model an aggregator with a queueing network.  We first
define the settings of the model from the system discussed in
Section~\ref{sec:overview} and give some assumptions.  Then, we construct a
queueing network, which is used to estimate the RU and RD capacities.

\subsection{Settings}

The V2G system is modelled as a queueing network with three queues, namely, the \textit{regulation-down queue} (RDQ), \textit{regulation-up-and-down queue}
(RUDQ), and \textit{regulation-up queue} (RUQ).  When an EV is plugged in at
time~$t$, the decision to join which queue depends on its SOC $x_i(t)$.  If it
is in States 1, 2, and 3 (defined in Section~\ref{sec:overview}) at time~$t$,
it will join the RDQ, RUDQ, and RUQ, respectively.  After joining a particular
queue, the following will happen:

\subsubsection{RDQ}

Each EV $i$ in this queue is actively charged at its own normalized charging
rate $r_i$.  If its SOC reaches $\underline{x}_i$ at time~$t'$, i.e.,
$x_i(t') = \underline{x}_i$, it will leave RDQ and join RUDQ.  The duration is
determined by:
\begin{align}
\Delta t = t' - t
	 = \frac{\underline{x}_i - x_i(t)}{r_i}, \quad x_i(t) < \underline{x}_i.
\label{q1time}
\end{align}
When an EV is actively charged, it gets served in the queue.  It is also
possible for it to depart from the queue before its SOC has reached
$\underline{x}_i$.  This represents the situation that it quits the system. 

\subsubsection{RUDQ}

When EV~$i$ arrives at this queue, it will be actively charged at the
normalized charging rate $r_i$ until its SOC reaches $\overline{x}_i$.  If the
charging process starts at time~$t$ and the EV is charged to $\overline{x}_i$
at time~$t'$, the duration is given by:
\begin{align}
\Delta t = t' - t = \frac{\overline{x}_i -x_i(t)}{r_i}, \quad
\underline{x}_i < x_i(t) < \overline{x}_i.
\label{q2time1}
\end{align}
If the EV joins from RDQ, we have:
\begin{align}
\Delta t = t' - t = \frac{\overline{x}_i - \underline{x}_i}{r_i}.
\label{q2time2}
\end{align}
After charging up to $\overline{x}_i$, the EV departs from this queue and goes
to RUQ.  Similar to RDQ, a departure of an EV from the queue before its SOC
reaching $\overline{x}_i$ corresponds to an EV leaving the system.

\subsubsection{RUQ}

When an EV joins this queue, no active charging takes place.  It will stay in
this queue until it departs from the system.
 
\subsection{Assumptions} \label{subsec:assumptions}

We make the following assumptions to make the analysis mathematically tractable:
\begin{enumerate}
\item The events associated with each EV and among EVs are independent with
      each other.  Each EV arrives at the system randomly, following a
      Poisson process at rate $\lambda$.  Among the EV arrivals, fractions $p_1$, $p_2$,
      and $p_3$ of EVs are in States 1, 2, and 3, respectively, where
      $p_1, p_2, p_3 \in [0, 1]$ and $p_1 + p_2 + p_3 = 1$.

\item There exists a smart charging mechanism
      $M_{SC}:(x_i(t), \underline{x}_i, \overline{x}_i) \mapsto r_i$, which
      assigns the normalized charging rate $r_i$ to EV~$i$ according to its
      current SOC $x_i(t)$ upon its arrival at time~$t$, and its target SOC 
      thresholds $\underline{x}_i$ and $\overline{x}_i$.  With such mechanism, the durations of EVs
      in States 1 and 2 (refer to \eqref{q1time}, and \eqref{q2time1} and
      \eqref{q2time2}, respectively) are exponentially distributed at rates
      $\mu_1$ and $\mu_2$, respectively.

\item There exists a fraction $q_1$ of EVs in State 1 which will directly quit
      the system.  This fraction represents those EVs whose SOCs do not reach
      their lower target SOC limits at their departures from the system.
      Similarly, we have a fraction $q_2$ of EVs which depart from the system in
      State 2.  Note that $q_1$ and $q_2$ already capture those EVs which
      require fast charging.  In other words, they may only stay in the system
      for a short period of time.

\item When an EV is in State 3, no charging would happen. It will remain on
      standby in the system for a period  exponentially distributed with
      rate $\mu_3$.
\end{enumerate}
  
The values of $\lambda$, $p_1$, $p_2$, $q_1$, and $q_2$ can be
determined by statistical measurements from the operations of the charging
facilities.  The smart charging mechanism $M_{SC}$ can help maintain exponential service times in States 1 and 2 and a modified $M_{SC}$ will be adopted to maintain exponential service times in States 3. We will discuss the design of (modified) $M_{SC}$ in Section \ref{sec:scm}.

\subsection{Model} \label{subsec:model}

\begin{figure}[!t]
\centering
\includegraphics[width=2.6in]{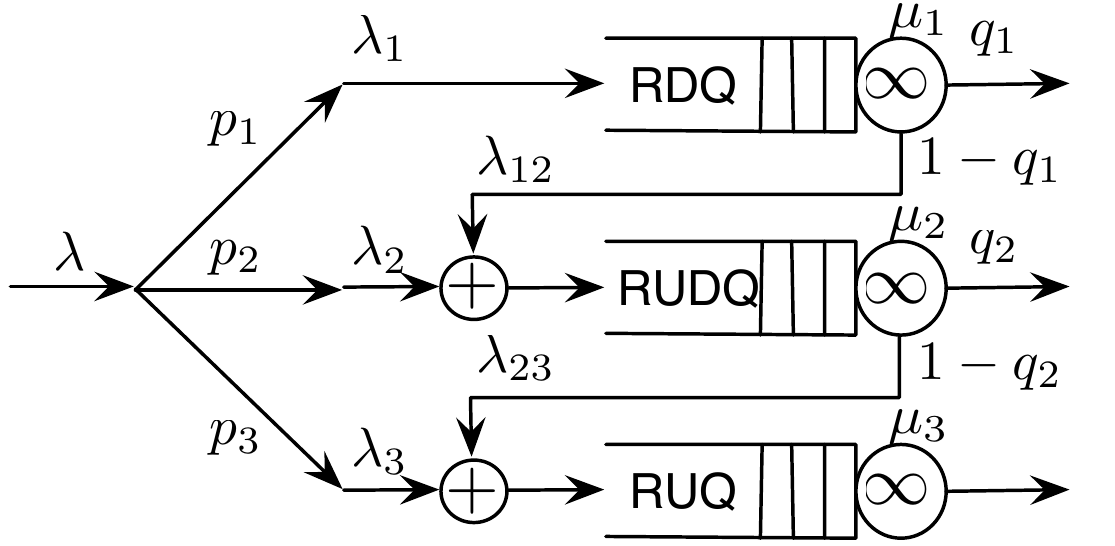} 
\caption{The queueing model.}
\label{fig:t_queues}
\end{figure}

Fig.~\ref{fig:t_queues} depicts the queueing model, where RDQ, RUDQ, and RUQ
model the behaviors of EVs in States 1, 2, and 3, respectively.  Assumption~1
states that we randomly split the EV arrival process into three subprocesses
according to the probability distribution $(p_1, p_2, p_3)$.  Since random
splitting results in independent Poisson subprocesses, the external arrivals
at each queue constitute a Poisson process with rate $\lambda_1$ for RDQ,
$\lambda_2$ for RUDQ, and $\lambda_3$ for RUQ, where $\lambda_1 = p_1 \lambda$,
$\lambda_2 = p_2 \lambda$, and $\lambda_3 = p_3 \lambda$.

When an EV enters RDQ, active charging starts immediately.  In other words, all
EVs in this queue can get served without queueing up.  With Assumption~2, an
EV resides in this queue with a duration exponentially distributed with rate
$\mu_1$.  Hence, RDQ can be modelled as an $M/M/\infty$ queue with arrival rate
$\lambda_1$ and service rate $\mu_1$.  According to \cite{queueingtheory}, the
probability $p_{1, n}$ of having $n$ EVs in this queue in the steady state is:
\begin{align}
p_{1,n} = \frac{(\frac{\lambda_1}{\mu_1})^n e^{-(\frac{\lambda_1}{\mu_1})}}{n!}.
\end{align}

The expected number $L_1$ of EVs charging in RDQ is:
\begin{align}
L_1 = \sum_{n = 1}^{\infty}{np_{1,n}}
    = \frac{\lambda_1}{\mu_1} = \frac{p_1 \lambda}{\mu_1}.
\label{L_1}
\end{align} 

By Burke's theorem~\cite{datanetworks}, the departure process of RDQ is a
Poisson process with rate $\lambda_1$.  With Assumption~3, this Poisson process
is split randomly according to the probability distribution $(q_1, 1 - q_1)$.
$(1 - q_1)$ of EVs enter RUDQ with rate $\lambda_{12} = (1 - q_1) \lambda_1$,
which superposes with the Poisson subprocess for the external arrivals with
rate $\lambda_2$.  Since the superposition of Poisson processes is still a
Poisson process, the combined arrivals to RUDQ constitute a Poisson process with
rate $(\lambda_2 + \lambda_{12})$.  Similar to RDQ, RUDQ can also be modelled
as an $M/M/\infty$ queue with arrival rate $(\lambda_2 + \lambda_{12})$ and
 service rate $\mu_2$.  Hence, the probability $p_{2, n}$ of having $n$ EVs
in this queue is:
\begin{align}
p_{2,n} = \frac{(\frac{\lambda_2 + \lambda_{12}}{\mu_2})^n
	  e^{- \frac{\lambda_2 + \lambda_{12}}{\mu_2}}}{n!}.
\end{align}

The expected number $L_2$ of EVs charging in RUDQ is given by:
\begin{align}
L_2 = \frac{\lambda_2 +\lambda_{12}}{\mu_2}
    = \frac{\lambda (p_1 + p_2 - p_1 q_1)}{\mu_2}.
\label{L_2}
\end{align} 

With Assumption~3, the departure process is Poisson with rate
$(\lambda_2 + \lambda_{12})$, which is split randomly according to the
probability distribution $(q_2, 1 - q_2)$.  $(1 - q_2)$ of EVs enter RUQ as
 a Poisson process with rate
$\lambda_{23} = (\lambda_2 + \lambda_{12}) \cdot (1 - q_2)$ for RUQ.  The
combined arrival process of RUQ is also a Poisson process with rate
$(\lambda_3 + \lambda_{23})$.  With Assumption~4, RUQ can be modelled as an
$M/M/\infty$ queue with arrival rate $(\lambda_3 + \lambda_{23})$ and service
rate $\mu_3$.  Therefore, the probability $p_{3, n}$ of having $n$ EVs in this
queue is:
\begin{align}
p_{3,n} = \frac{(\frac{\lambda_3 + \lambda_{23}}{\mu_3})^n
	  e^{- \frac{\lambda_3 + \lambda_{23}}{\mu_3}}}{n!}.
\end{align}

The expected number $L_3$ of EVs standing by in RUQ can be expressed as:
{\small
\begin{align}
L_3 = \frac{\lambda_3 + \lambda_{23}}{\mu_3}
    = \frac{\lambda (1 - p_1 q_1 - p_1 q_2 - p_2 q_2 + p_1 q_1 q_2)}{\mu_3}.
\label{L_3}
\end{align} 
}

The overall system departure process is a Poisson process superposed by three
individual departure Poisson processes from the three queues.  The overall
departure process has rate:
\begin{align}
\lambda = q_1 \lambda_1 + q_2 (\lambda_2 + \lambda_{12}) +
	  (\lambda_3 + \lambda_{23}).
\end{align}

The duration of each regulation service $\Delta t_{reg}$ is normally short,
such as a few minutes~\cite{frequencyregulation}, while EVs are expected to
switch their states in a relatively much lower rate.  Thus, the mean service
times of the queues, $\frac{1}{\mu_1}$, $\frac{1}{\mu_2}$, and
$\frac{1}{\mu_3}$ are generally much longer than a few minutes.  This is
justifiable as an EV cannot be charged up nor leave the system within a few
minutes on the average.  For each EV, the amount of power $P_{EV}$ contributed
for a regulation event can be determined with the amount of energy required
$\Delta x_{EV}$ by $P_{EV} = \frac{\Delta x_{EV}}{\Delta t_{reg}}.$

As an aggregator normally coordinates hundreds of EVs, $P_{EV}$ contributed by
a single EV is small.  Hence, $\Delta x_{EV}$ would be even smaller.  For a
particular regulation contract with the fixed regulation service duration
$\Delta t_{reg}$, we can fix $P_{EV}$ to be small enough such that the
probability of having a state transition of an EV after an absorption or a
removal of energy of $\Delta x_{EV}$ for a regulation service is almost
negligible.\footnote{Interested readers can refer to Section~\ref{sec:perf}
for performance results.}  Therefore, the capacities for the regulation
services can be estimated based on the numbers of EVs available for regulation.
Due to the types of regulation supported by EVs as described in
Section~\ref{sec:overview}, the steady state RD capacity 
$C_{RD}$ can be computed as:
\begin{align}
C_{RD} = P_{EV}(L_1 + L_2).
\label{Crd}
\end{align}
Similarly, the steady state RU capacity $C_{RU}$ is given by:
\begin{align}
C_{RU} = P_{EV}(L_2 + L_3).
\label{Cru}
\end{align}

\section{Smart Charging Mechanism} \label{sec:scm}

Recall that in Section \ref{subsec:assumptions}, a smart charging mechanism $M_{SC}$ is adopted to assign charging rates to  EVs based on their battery statuses so that their durations in States 1 and 2 follow exponential distributions. We also apply a modified $M_{SC}$ to the EVs in State 3 such that their durations in the system follow an exponential distribution as well. 
However, this $M_{SC}$ is posed as an assumption when the model is developed in Section \ref{sec:model}. The model will not work and the resultant capacity estimation cannot be validated if $M_{SC}$ does not exist. To complete the model, in this section, we discuss the design of such a smart charging mechanism.

\begin{figure}[!t]
\centering
\includegraphics[width=3.5in]{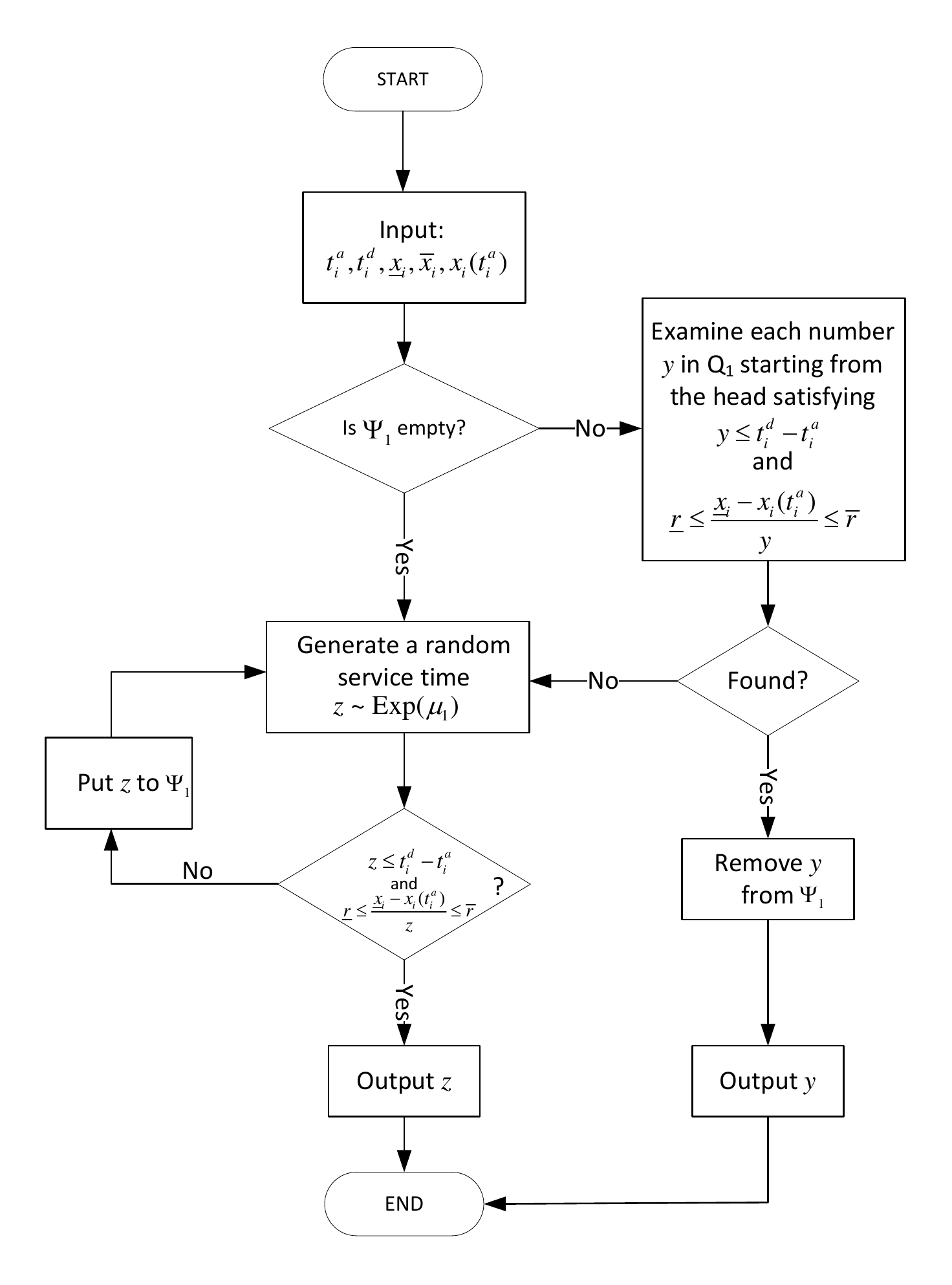} 
\caption{Flow chart of the smart charging mechanism for RDQ.}
\label{fig:scm}
\end{figure}
When EV $i$ arrives at the system, it specifies its arrival time $t_i^a$, its expected departure time $t_i^d$, its lower and upper SOC target thresholds $\underline{x}_i$ and $\overline{x}_i$, and its initial SOC $x_i(t_i^a)$.
The main purpose of the smart charging mechanism is to assign a service
duration $w_i$ to EV $i$ such that:
\begin{enumerate}
\item the service times assigned to a set of EVs statistically follow the exponential distribution;
\item the service time should not be longer than the expected duration of the
EV staying in the system; and
\item the required charging rate should fall into the range of $[\underline{r},\overline{r}]$, where $\underline{r}$ and $\overline{r}$ are the lower and upper charging rate limits supported by the system.
\end{enumerate}
Point 1 ensures that an analytical model for capacity characterization can be developed based on Section \ref{subsec:model}. Point 2 guarantees that the EVs will not stay longer in the system than expected.  That is, $w_i$ should be smaller than or equal to $t_i^d-t_i^a$. For Point 3, when the EVs stay in RDQ and RUDQ, their batteries need to be charged up to certain levels, and thus, the charging rates should be feasible for the system.
For a particular queue, $M_{SC}$ aims to assign the arriving EVs with service durations, which statistically follow an exponential distribution. Upon an arrival, $M_{SC}$ will assign a service duration to the EV. In general, not every randomly generated service duration fits the condition of the EV. However, we can reserve any unfit random service duration for another EV which comes at a later time. We introduce Lemma \ref{lemma:permutation} and Corollary \ref{col:realization} to formally provide the underlying mathematical reasoning.
\begin{lemma} \label{lemma:permutation}
Consider that the sequence $y(t)=[y_1, y_2,\ldots,y_t]$ is a realization of the independent and identically distributed (i.i.d.) random process $\{Y_i\}$. Let $\sigma$ be any permutation of the indices $1, 2, \ldots, t$. $y'(t)=[y_{\delta(1)}, y_{\delta(2)}, \ldots,y_{\delta(t)}]$ is also a realization of $\{Y_i\}$.
\end{lemma}
\begin{proof}
Regardless of the order, $y(t)$ and $y'(t)$ contain the same set of numbers $\{y_1,\ldots,y_t\}$. As $\{Y_i\}$ is i.i.d., $y'(t)$ is also a realization of $\{Y_i\}$.
\end{proof}
\begin{corollary} \label{col:realization}
Consider that $y(t)=[y_1, y_2,\ldots,y_t]$ is a realization of the i.i.d. $\{Y_i\}$.
For $1\leq k\leq t'\leq t$, $y'(t)=[w_1, \ldots,y_{k-1},y_{k+1},\ldots,y_{t'},y_k,y_{t'+1},\ldots,y_{t}]$ is also a realization of $\{Y_i\}$.
\end{corollary}

Since RDQ, RUDQ, and RUQ have different expected service times and SOC charging ranges (see Section \ref{sec:model}), we apply variants of the smart charging mechanism for the three queues and their design principles are largely similar. We will explain the detailed design for RDQ and then point out the differences for RUDQ and RUQ:

\subsubsection{RDQ}
The design principle is to set the charging rates $r_i$'s of the EVs by assigning with corresponding $w_i$'s, where $w_i$'s follow the exponential distribution with mean $\mu_1$ as much as possible.   
Consider a sequence of exponentially distributed numbers $[y_1,y_2,\ldots]$ with mean $\mu_1$, and they are potential service times of  EVs in State 1. The idea is that, when an EV comes, we assign the $y_j$ with the smallest index $j$ to its service time such that $y_j$ fits its specifications. Once $y_j$ has been adopted, we remove it from the sequence.
Fig. \ref{fig:scm} shows the implementation details. We maintain a first-in-first-out queue $\Psi_1$ to temporarily store the previously unadopted random numbers for $y$.
Suppose that EV $i$ with $t_i^a$, $t_i^d$, $\underline{x}_i$, $\overline{x}_i$, and $x_i(t_i^a)$ enters the queue. We first check if there are any unassigned numbers in $\Psi_1$. If so, we examine each number $y$ in $\Psi_1$ to see if $y$ satisfies 
\begin{align}
y\leq t_i^d - t_i^a
\label{eq:durationCons}
\end{align}
and
\begin{align}
\underline{r}\leq \frac{\underline{x}_i - x_i(t_i^a)}{y} \leq \overline{r}.
\label{eq:rateCons}
\end{align}
If there are multiple qualified $y$'s, we select the $y$ which was the earliest to be generated. Then we remove $y$ from $\Psi_1$ and set EV $i$'s service time $w_i = y$. 
Condition \eqref{eq:durationCons} means that the service time $w_i$ will not be longer than the EV's expected duration of stay, i.e., $t_i^d - t_i^a$. Since  EV $i$ in this queue will be charged up to $\underline{x}_i$,  the amount of energy will be charged  is $\underline{x}_i - x_i(t_i^a)$, and thus, the required charging rate is $r_i=\frac{\underline{x}_i - x_i(t_i^a)}{y}$.  Condition \eqref{eq:rateCons} ensures that $r_i$ can be supported by the system.
We examine from the head of the queue to ensure that the earliest qualified number in the queue can be adopted first. 
If there is no qualified $y$ in $\Psi_1$,  we keep generating random numbers $z$'s which are exponentially distributed with mean $\mu_1$ until we get one $z$ satisfying \eqref{eq:durationCons} and \eqref{eq:rateCons}. All those unqualified numbers will be queued up in $\Psi_1$.
Note that a random service time $y$ which violates the conditions of a particular EV, i.e., \eqref{eq:durationCons} and \eqref{eq:rateCons}, may be qualified with another EV.  Our goal is to keep the length of $\Psi_1$ at all times as short as possible. In this way, we just postpone some random numbers used at a later time. Here we just repeatedly apply Corollary \ref{col:realization} to design $M_{SC}$.

\subsubsection{RUDQ}
The $M_{SC}$ design for RUDQ is similar to that for RDQ but the random numbers should follow exponential distribution with mean $\mu_2$ as much as possible. All EVs will be charged up to their $\overline{x}_i$'s. If EV $i$ comes from RDQ,  the amount of energy needed to be charged is $\overline{x}_i - \underline{x}_i$. Otherwise, the required charged amount is $\overline{x}_i - x_i(t_i^a)$.  We maintain a queue $\Psi_2$ to store the unqualified random numbers.

\subsubsection{RUQ}
We modify the previously discussed $M_{SC}$ for RUQ. Instead of setting charging rates, we manipulate the service times for the EVs only. The design principle  is also similar to that for RDQ but the random numbers should seemingly follow exponential distribution with mean $\mu_3$. All EVs in RUQ do not need to be actively charged, and thus, the rate constraint like \eqref{eq:rateCons} is not required. We maintain a queue $\Psi_3$ to store the unqualified random numbers. 
When an EV is assigned with a service time shorter than its own expected duration of stay, it can still physically park at the parking infrastructure but we just disconnect it from the system. 

We need a performance metric to determine how well $M_{SC}$ facilitates the analytical model discussed in Section \ref{subsec:model}. We have Theorem \ref{thmresult} to introduce the queue length of $\Psi_1$ to evaluate the performance of $M_{SC}$. It says that if the queue length remains finite, the service durations assigned to the EVs will follow an exponential distribution in the long run. The details can be found as follows:
\begin{theorem} \label{thmresult}
Let $N(t)$ be the queue length of $\Psi_1$ at time $t$, and $y(t)=[y_1,\ldots,y_t]$ and $w(t)=[w_1,\ldots,w_t]$ be the random number sequence generated from an exponential distribution with mean $\mu_1$ and the set of qualified service times adopted by the first $t$ participating EVs, respectively, at RDQ.
 If $t$ tends to infinity and $\lim_{t\rightarrow \infty}N(t)<\infty$
holds, $w(t)$ is exponentially distributed with mean $\mu_1$ almost surely.
\end{theorem}
\begin{proof}
$y_k$'s in $y(t)$ are split into two sets, either becoming $w_k$'s in $w(t)$ or being stored in $\Psi_1$.
The condition $\lim_{t\rightarrow \infty}N(t)<\infty$ implies that only a finite number of $ y_k$'s go to $\Psi_1$ and thus $w(t)$ must contain an infinite number of $w_k$'s. When being assigned to $w(t)$, some $y_k$'s may have been re-ordered;  $y_k$ may have been assigned to $w_{k'}$, where $k'>k$. By Corollary  \ref{col:realization}, permutations of $y_k$'s in the sequence does not affect the exponential nature of the sequence. Although $y(t)$ and $w(t)$ are different, $w(t)$ has inherited $y(t)$'s exponential distribution property almost surely.
\end{proof}

Similar results also hold for RUDQ and RUQ.
Note that we do not need to specify any requirement of the inputs; the initial SOCs $x_i(t_i^a)$, the duration of stays \ $t_i^d - t_i^a$, and the lower and upper SOC target thresholds $\underline{x}_i$ and $\overline{x}_i$ can follow \textit{unknown} distributions. The smart charging mechanism can adapt to the characteristics of the inputs and generate exponentially distributed service times with means specified for the queues. 
\section{Performance Evaluation} \label{sec:perf}

\subsection{System Parameter Settings} \label{subsec:parameter}
We study the performance of the system with a parking structure example, where EVs arrive
and leave independently.  Consider a scenario that there are five EVs entering the
parking structure per minute following a Poisson distribution on the average.  90\% of EVs require charging, where their SOCs
are below their upper target thresholds at their arrivals.  One tenth of them
do not, since they need parking only and their SOCs are above their respective
upper target thresholds.  Among those requiring a charge, based on \cite{SOCdistribution}, we assume their initial SOCs $x_i(t_i^a)$'s follow a truncated Normal distribution in the range of [0, 1] with mean 0.5 and standard deviation 0.2. The EVs expect to be
charged up to their SOC upper target thresholds $\overline{x}_i$'s, each of which is set with a truncated Normal distribution in the range of $[x_i(t_i^a),1]$ with mean $x_i(t_i^a)+ 0.5 (1-x_i(t_i^a))$ and standard deviation $0.1 (1-x_i(t_i^a))$ for EV $i$. The SOC lower target threshold $\underline{x}_i$ of EV $i$ is set to $\underline{x}_i = \overline{x}_i\times \text{rand}[0.6,0.8]$, where $\text{rand}[0.6,0.8]$ is a random number uniformly generated in $[0.6,0.8]$. 
There are different charging standards and the required charging times for typical EV models vary from 20 minutes to 8 hours \cite{chargingDuration}. Moreover, the durations of parking also depend on the drivers' driving practice, the purposes of parking, and the initial SOCs. We assume that the durations of stay (i.e., $t_i^d -t_i^a$ for EV $i$) are normally distributed and truncated in the range of [60, 780] minutes with mean 420 minutes and standard deviation 60 minutes. With these settings, we classify EVs into states according to their SOCs and simulation gives $p_1 \approx 0.5$, $p_2 \approx 0.4$, and $p_3 = 0.1$. Based on the current charging technologies, we set the range $[\underline{r},\overline{r}]$ of the charging rates as [0, 0.05], where $\overline{r}$ corresponds to fast charging \cite{fastcharging}. We also have $q_1 = q_2 = 0.1$.
The above shows how to determine the system parameters, $\lambda$, $p_1$, $p_2$, $p_3$, $q_1$, and $q_2$, for our illustrative scenario. Note that we do not require any specific EV arrival rates and distributions for SOC conditions and total durations of stays. Different scenarios just give different parameter settings.

\subsection{Results for a Reference Set of $\mu_1$, $\mu_2$, and $\mu_3$} \label{subsec:mu123}
The only parameters that we can control are the expected service times of the EVs, in terms of $\mu_1$, $\mu_2$, and $\mu_3$, at the three queues by setting corresponding charging rates and/or disconnecting EVs from the system. Smaller values of $\mu_1$, $\mu_2$, and $\mu_3$ give larger regulation capacities but inappropriate values may result in remarkable errors between the analytical and actual capacities. We first define a reference set of $\mu_1 = \frac{1}{50} \text{ min}^{-1}$, $\mu_2 = \frac{1}{70} \text{ min}^{-1}$, $\mu_3 = \frac{1}{30} \text{ min}^{-1}$. In this way, EVs in State~1 will spend 50~minutes in this state
on the average. When they exit State 1, 10\% of them leave the system, while the
 other 90\% of them transit to State 2 and continue to  charge up their
batteries to their upper target thresholds  with a mean service
time of 70~minutes.  Next 10\% of them leave the system from State 2. 
and the rest of the EVs stay in State~3 (without
charging) with the mean residence time equal to 30~minutes. 
By \eqref{L_1}, \eqref{L_2}, and \eqref{L_3}, the expected numbers of EVs in
States~1-3 are $L_1 = 127.32$, $L_2 = 296.55$, and $L_3 = 129.65$ in the steady
state.

We set $P_{EV} = 6$~kW and $\Delta t_{reg} = 1$~min.  Thus, each EV absorbs
or delivers $0.1\text{ kWh}$ for each regulation service.  When compared with
the EV models already available in the market, such small charging and
discharging rates of $P_{EV}$ result in a regulation event where the involved
EVs do not switch states when supporting regulation (i.e., no transition to
other queues merely for regulation).  For example, the Tesla Model~S has a
battery capacity ranging from 40~kWh to 85~kWh~\cite{tesla}, and BYD~e6 has a
battery capacity of 60~kWh~\cite{byd}.  With \eqref{Crd} and \eqref{Cru}, the
expected RD and RU capacities are $C_{RD} = 6\text{ kW} \times (127.32 + 296.55) = 2543.22\text{ kW}$ and $C_{RU} = 6\text{ kW} \times (296.55 + 129.65) = 2557.19\text{ kW}$, respectively.

\begin{figure}[!t]
	\begin{center}
		\subfigure[RDQ]{\label{fig:Q1cap}\includegraphics[width=3.5in]{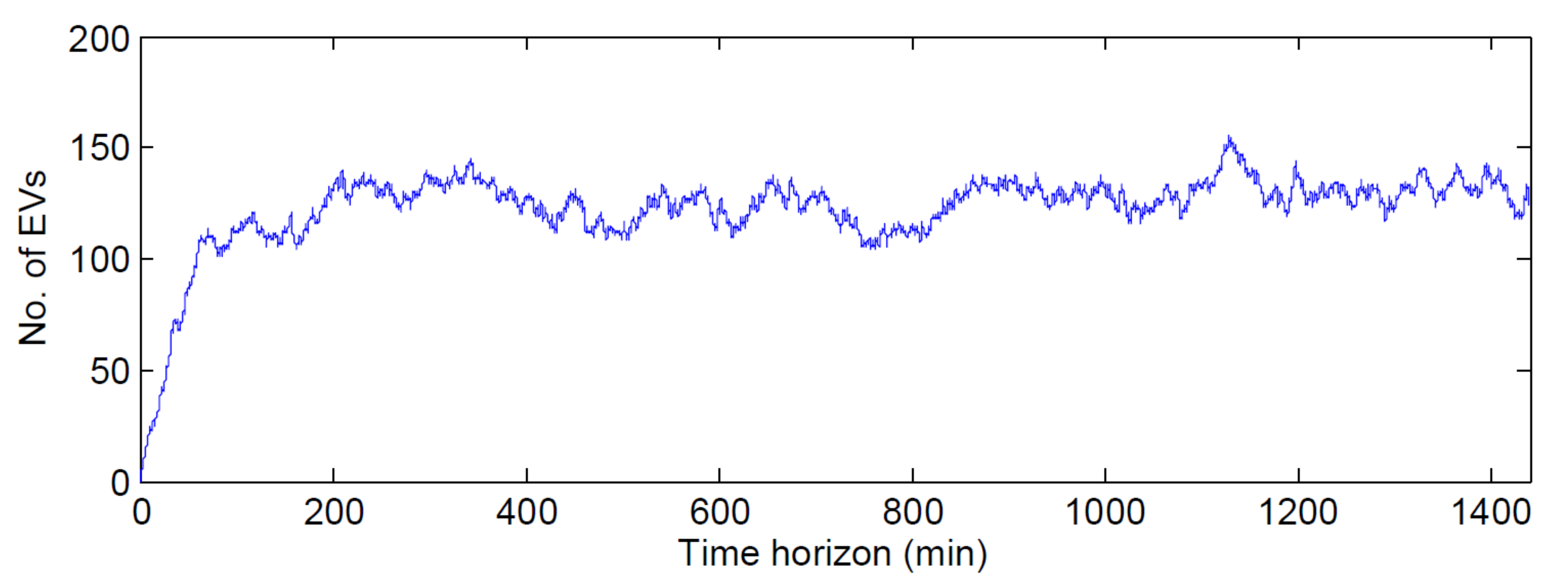}}
    	\subfigure[RUDQ]{\label{fig:Q2cap}\includegraphics[width=3.5in]{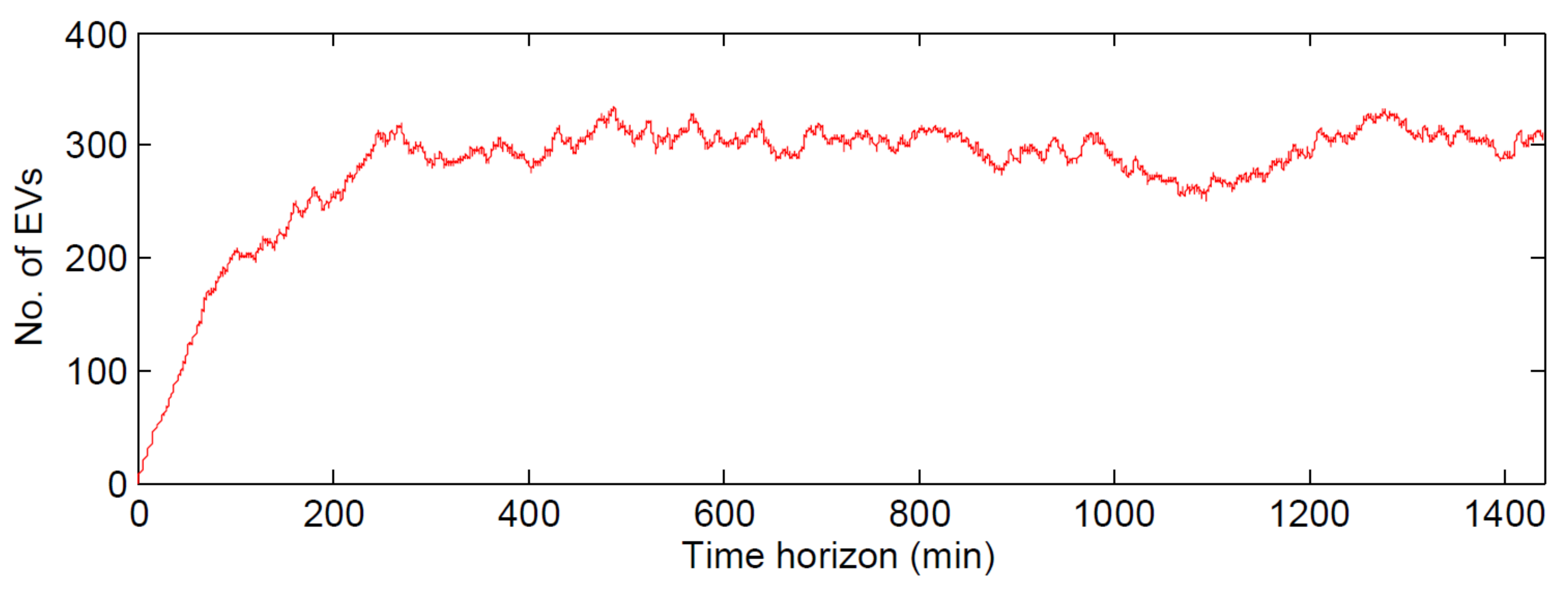}} 
    	\subfigure[RUQ]{\label{fig:Q3cap}\includegraphics[width=3.5in]{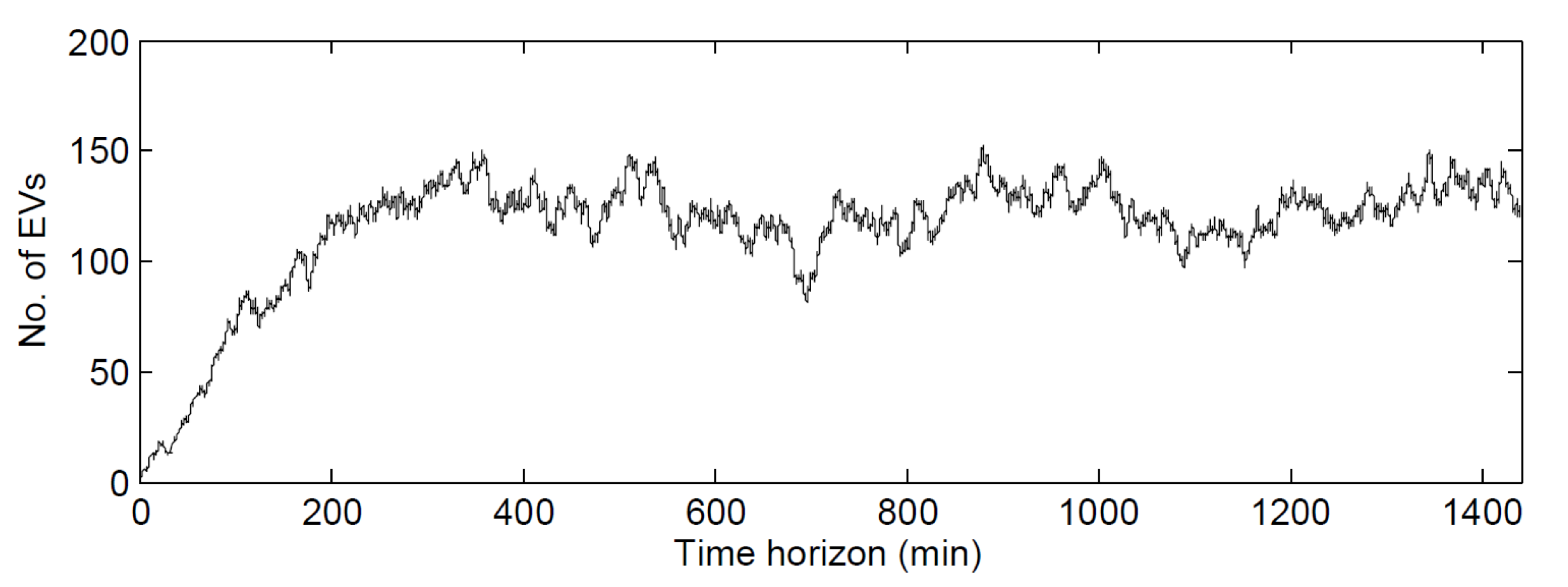}}
	\end{center}
	\caption{Variations of the number of EVs in each queue.}
  \label{fig:t_sim}
\end{figure}

We simulate the instantaneous numbers of EVs getting served in the queues for
1440~minutes.  The simulation was implemented with Matlab.  The results
are exhibited in Fig.~\ref{fig:t_sim}.  The system is initially empty and it
takes about 200~minutes to reach the steady state, where the numbers of EVs in
the queues oscillate around our computed expected values.

Fig. \ref{fig:scm_queue} shows the queue length variations of $\Psi_1$, $\Psi_2$, and $\Psi_3$ for $M_{SC}$ generated from the same simulation run as in Fig. \ref{fig:t_sim}. Fig. \ref{fig:scmQ1} illustrates the variations of queue length of $\Psi_1$ corresponding to the variations of EV population size shown in Fig. \ref{fig:Q1cap}. Similarly, Figs. \ref{fig:scmQ2} and \ref{fig:scmQ3} correspond to Figs. \ref{fig:Q2cap} and \ref{fig:Q3cap}. In spite of the existence of some blips, the sizes of $\Psi_1$, $\Psi_2$, and $\Psi_3$ do not grow continuously in the time horizon. 
By Theorem \ref{thmresult}, the actual system supported by the smart charging mechanism will generally follow the analytical results developed in Section \ref{subsec:model} in the long run.
$\Psi_1$ and $\Psi_3$ can return to empty queues from time to time and this implies that all their unqualified random numbers for service times at particular instants can be adopted in their next few instants. The length of $\Psi_2$ is stabilized at around 15 and it is likely that the adoption rate and the generation rate of unqualified random numbers are similar.

\begin{figure}[!t]
	\begin{center}
		\subfigure[$\Psi_1$]{\label{fig:scmQ1}\includegraphics[width=3.5in]{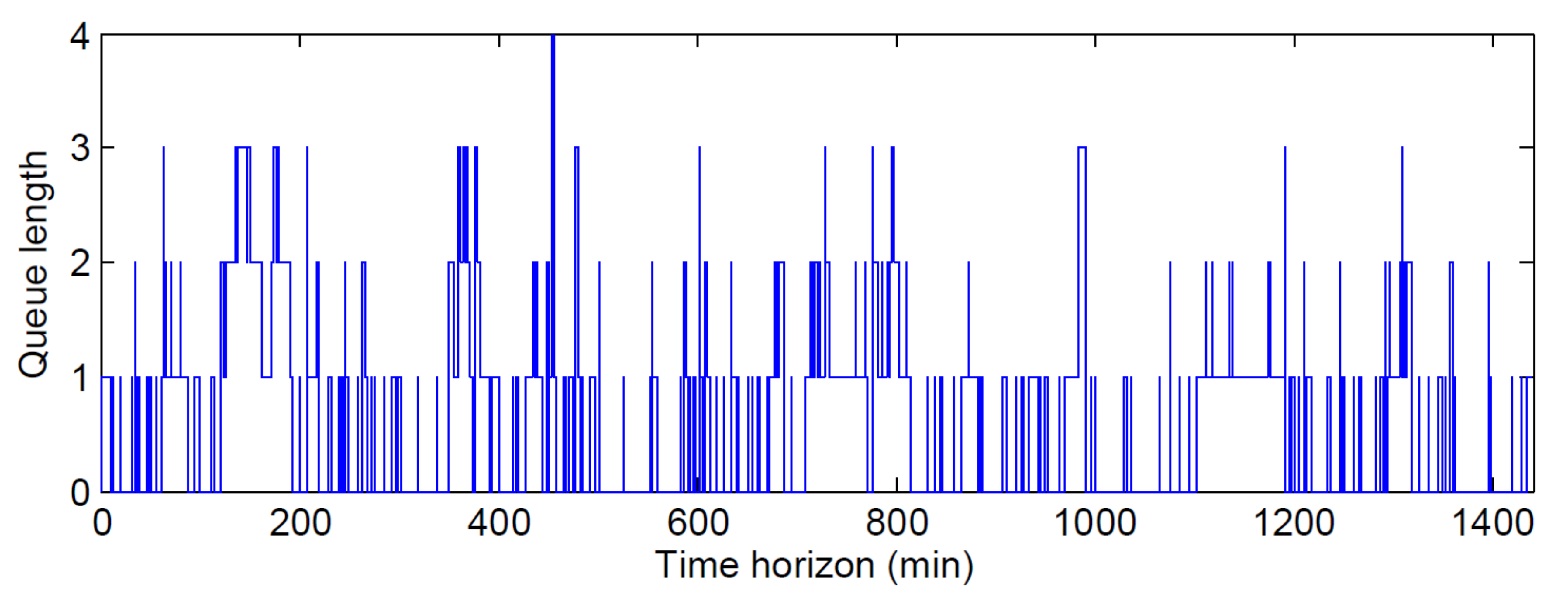}}
    	\subfigure[$\Psi_2$]{\label{fig:scmQ2}\includegraphics[width=3.5in]{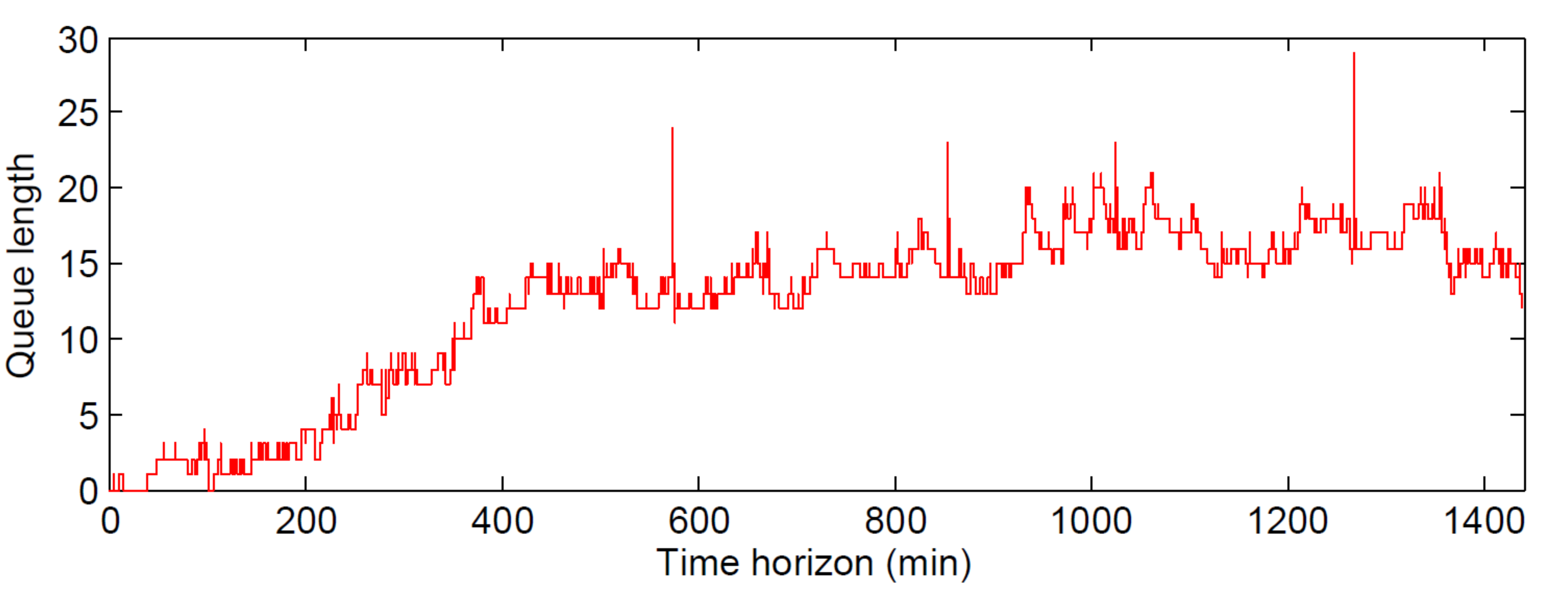}} 
    	\subfigure[$\Psi_3$]{\label{fig:scmQ3}\includegraphics[width=3.5in]{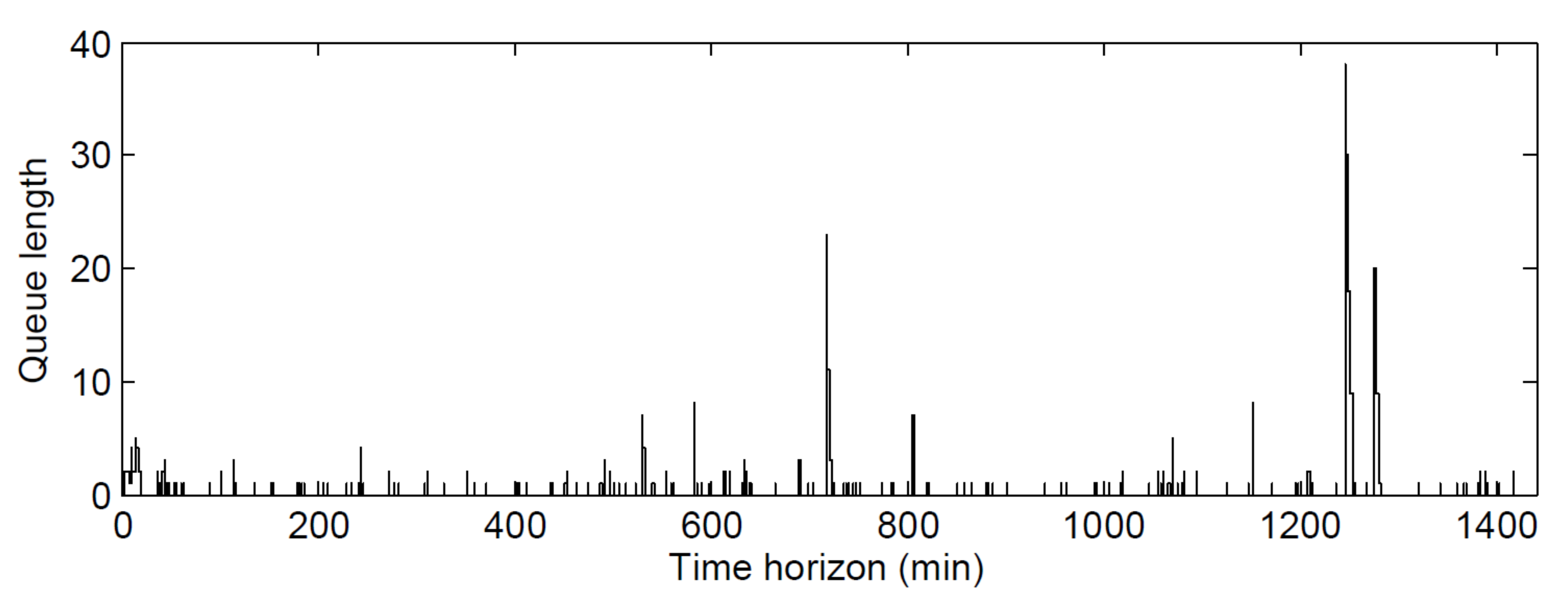}}
	\end{center}
	\caption{Variations of the queue lengths for $M_{SC}$.}
  \label{fig:scm_queue}
\end{figure}

\subsection{Effects of Different Values of $\mu_1$, $\mu_2$, and $\mu_3$}
Here we examine how the values of $\mu_1$, $\mu_2$, and $\mu_3$ affect the system performance. We consider  the combinations of $\mu_1 = \frac{1}{50} \text{ min}^{-1}$, $\mu_2=\frac{1}{70} \text{ min}^{-1}$, and $\mu_3=\frac{1}{30} \text{ min}^{-1}$ adopted in Section \ref{subsec:mu123} as a reference. We vary either $\mu_1$, $\mu_2$, or $\mu_3$ from the reference each time to investigate the changes in performance.

Fig. \ref{fig:scm_all} gives the simulated RU and RD capacities and the errors between analytical and simulated capacities for different combinations of values of $\mu_1$, $\mu_2$, or $\mu_3$, where each error is computed by $\frac{\text{simulated value - analytical value}}{\text{analytical value}}$.
In Fig.~\ref{fig:scm_all}, ``C'' and ``E'' stand for ``capacity'' and ``error'', respectively. Thus, for example, C\_RD represents the simulated RD capacity. Each data point is the average of 100 random cases generated from the settings as in Section \ref{subsec:parameter}. We change the values of $\mu_1$, $\mu_2$, and $\mu_3$, with respect to the reference, in Figs. \ref{fig:mu1}, \ref{fig:mu2}, and \ref{fig:mu3}, respectively. 
In Fig. \ref{fig:mu1}, the RD capacity increases with $\frac{1}{\mu_1}$, i.e., decreases with $\mu_1$. The RU capacity does not change as it does not relate to RDQ. In Fig. \ref{fig:mu2}, both capacities grow with $\frac{1}{\mu_2}$ since they both involve RUDQ. In Fig. \ref{fig:mu3}, the RU capacity increases with $\frac{1}{\mu_3}$ while the RD capacity becomes insensitive. It is because only the former is related to RUQ.
The trends of the error are similar. On the average, the simulated capacities are always smaller than the analytical ones and the discrepancies grow with $\frac{1}{\mu_1}$, $\frac{1}{\mu_2}$, and $\frac{1}{\mu_3}$. 
A certain part of the error is due to the transient state of the system as the analytical results only illustrate the steady state performance.
The rest comes from the growth of the queue lengths of $\Psi_1$, $\Psi_2$, and $\Psi_3$ for $M_{SC}$ because the probability of generating large and unqualified service times is higher with smaller $\mu_1$, $\mu_2$, or $\mu_3$. 
Therefore, there is a tradeoff between the resultant capacity and error. To construct the system with larger capacities, we need to tolerate larger errors. In practice, we set the values for $\mu_1$, $\mu_2$, and $\mu_3$ based on our requirements for the expected capacities and their accuracies.
Moreover, for a real parking structure, the characteristics of the arriving EVs may change with the time of day. Based on some historical data, we may divide a day into several periods, where the EV characteristics are more or less similar in each period. Then an appropriate combination of $\mu_1$, $\mu_2$, and $\mu_3$ can be set for each period and the capacities of each period can be deduced accordingly.

\begin{figure}[!t]
	\begin{center}
		\subfigure[Change in $\mu_1$]{\label{fig:mu1}\includegraphics[width=3.5in]{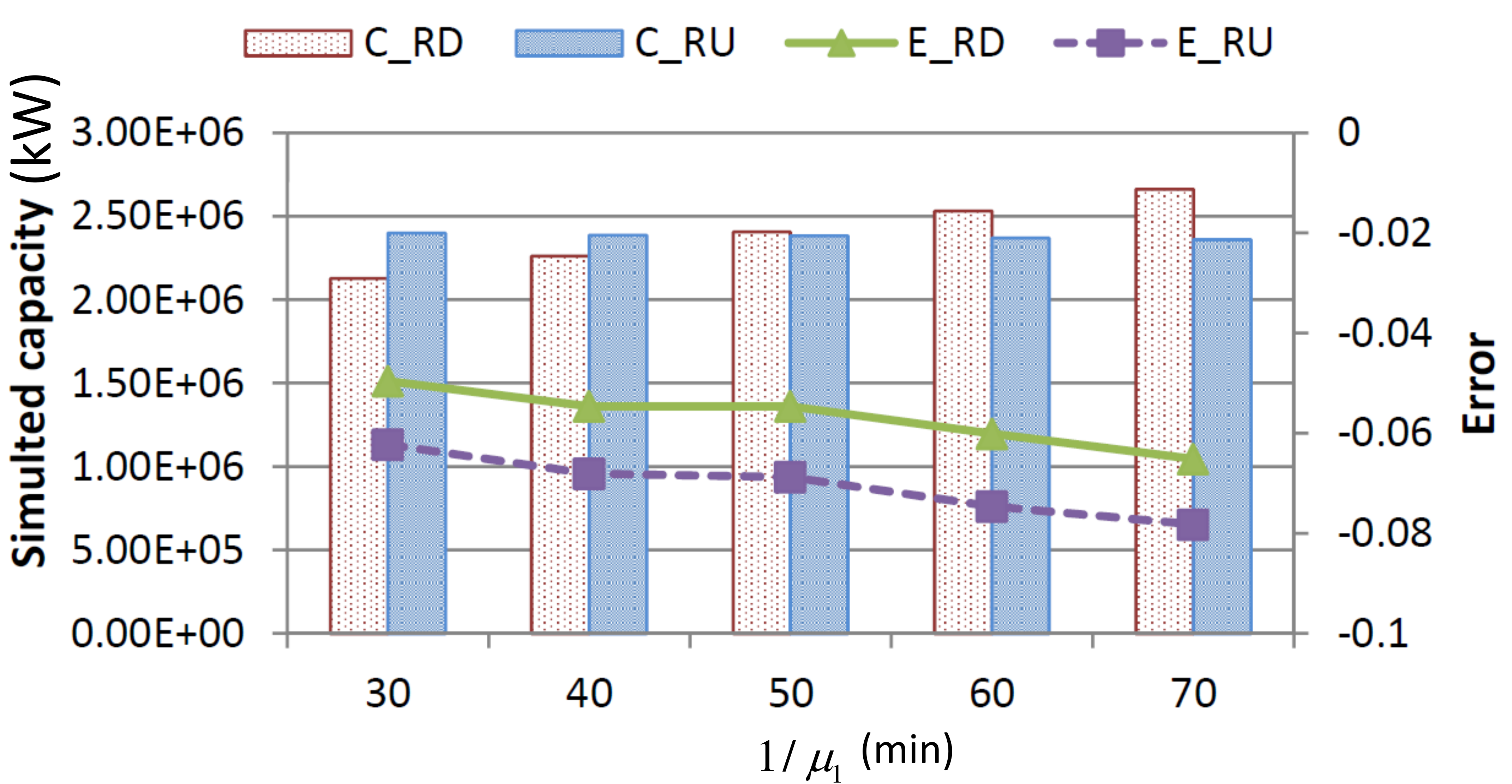}} 
    	\subfigure[Change in $\mu_2$]{\label{fig:mu2}\includegraphics[width=3.5in]{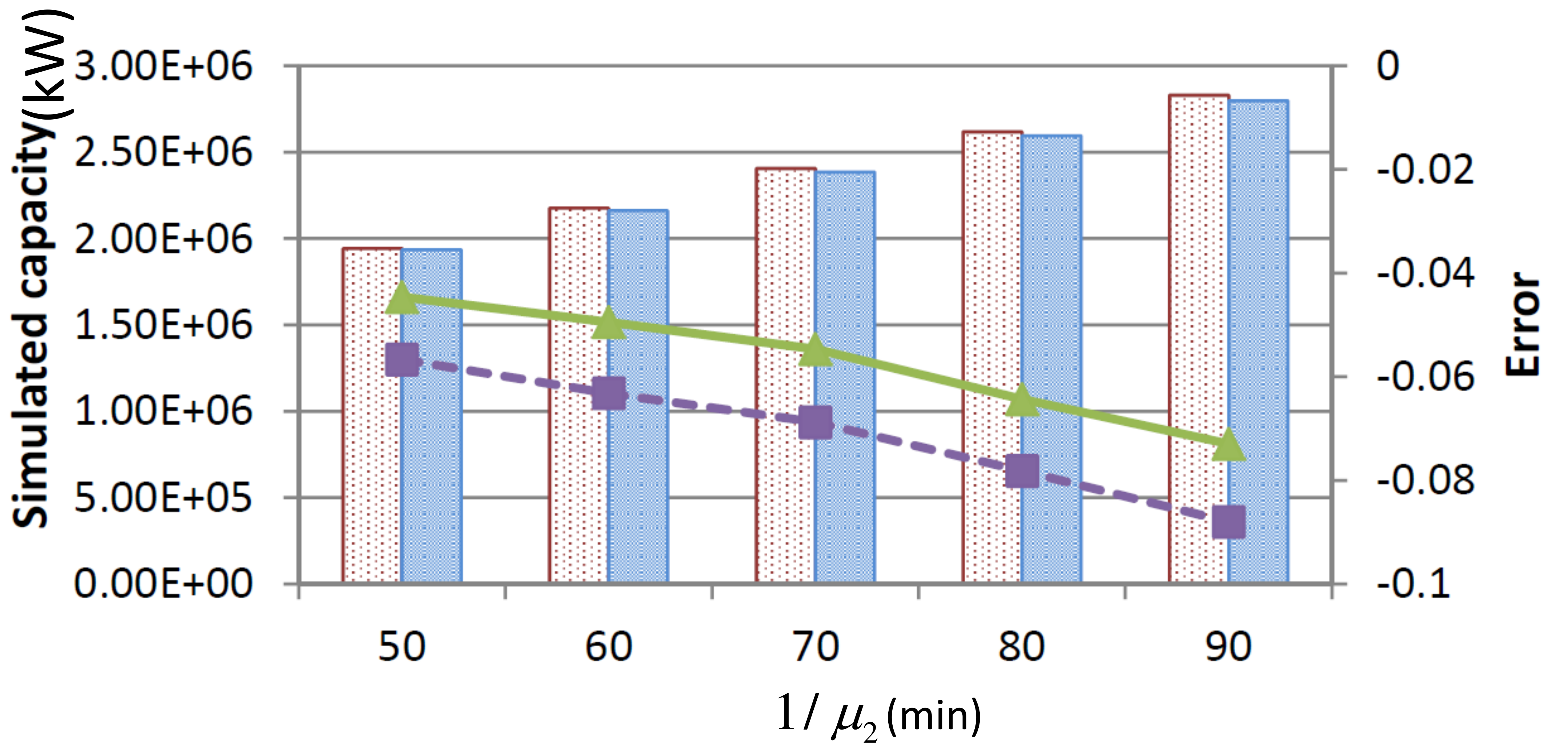}} 
    	\subfigure[Change in $\mu_3$]{\label{fig:mu3}\includegraphics[width=3.5in]{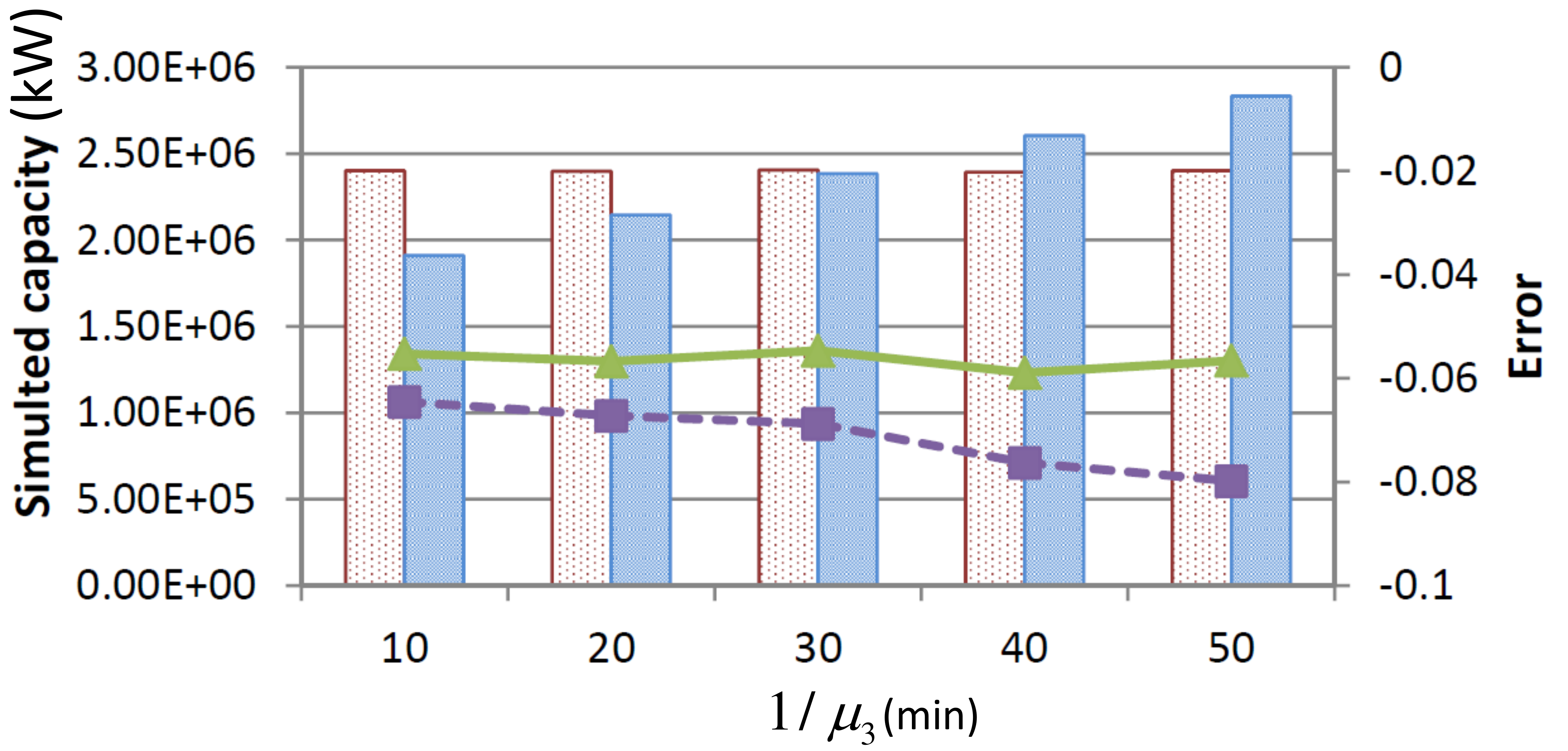}} 
	\end{center}
	\caption{Simulated capacities and the discrepancies between the analytical and simulated results.}
  \label{fig:scm_all}
\end{figure}

\section{Conclusion} \label{sec:concl}

With the expected higher penetration of renewable energy generation in smart grid, the
stochastic nature of the renewables will induce new challenges in matching the
actual power consumption and supply.  One measure to enforce power balance is
through regulation services.  
Traditional regulation services are mainly run by power plants and  very costly.  
The increasing social consensus on environmentally friendly transportation
 would lead to more reliance on EVs.  With
the embedded rechargeable batteries in EVs, a fleet of EVs can behave as a huge
energy buffer, absorbing excessive power from the smart grid or
supplying power to overcome the deficit.  This implies that an aggregation of EVs
is a practical alternative to support the regulation services of smart grid.
However, V2G is a dynamic system.  Each EV connects to and disconnects from the
system independently. Regulation through V2G can be realized only if we can capture the aggregate behavior of the EVs. However, in general, we cannot directly control the participating EVs and make them contribute according to our requirements. We can only estimate the collective contribution from the EVs while allowing them to behave autonomously. In this paper, we model an aggregation of EVs with a
queueing network.  The structure of the queueing network allows us to estimate
the RU and RD capacities separately. 
 The estimated capacities can help set up a regulation contract between an aggregator and a
grid operator so as to facilitate a new business model for V2G.
To make the results analytically tractable, the EV service durations need to be exponentially distributed. We introduce a smart charging mechanism to fulfill this requirement.
This mechanism does not require any specific patterns for the EVs' initial SOCs, the duration of stay, and their lower and upper SOC thresholds and it can adapt to the characteristics of the EVs and make the performance of the actual system follow the analytical model. 
To summarize, our contributions include: 1) proposing a queueing network model for V2G regulation services, 2) facilitating various regulation contracts by separating RU and RD capacities, 3) designing a smart charging mechanism to make the system adaptable to various characteristics of EVs, and 4) performing extensive simulations to verify the performance of the model.

\ifCLASSOPTIONcaptionsoff
  \newpage
\fi

\section*{Acknowledgment}
This research is supported in part by the Theme-based Research Scheme of the Research Grants Council of Hong Kong, under Grant No. T23-701/14-N. 

\bibliographystyle{IEEEtran}
\bibliography{IEEEabrv}

\begin{thebibliography}{10}
\providecommand{\url}[1]{#1}
\csname url@samestyle\endcsname
\providecommand{\newblock}{\relax}
\providecommand{\bibinfo}[2]{#2}
\providecommand{\BIBentrySTDinterwordspacing}{\spaceskip=0pt\relax}
\providecommand{\BIBentryALTinterwordstretchfactor}{4}
\providecommand{\BIBentryALTinterwordspacing}{\spaceskip=\fontdimen2\font plus
\BIBentryALTinterwordstretchfactor\fontdimen3\font minus
  \fontdimen4\font\relax}
\providecommand{\BIBforeignlanguage}[2]{{%
\expandafter\ifx\csname l@#1\endcsname\relax
\typeout{** WARNING: IEEEtran.bst: No hyphenation pattern has been}%
\typeout{** loaded for the language `#1'. Using the pattern for}%
\typeout{** the default language instead.}%
\else
\language=\csname l@#1\endcsname
\fi
#2}}
\providecommand{\BIBdecl}{\relax}
\BIBdecl

\bibitem{smartgridcomm2012}
A.~Y.~S. Lam, K.-C. Leung, and V.~O.~K. Li, ``Capacity management of
  vehicle-to-grid system for power regulation services,'' in \emph{Proc. of
  IEEE Int. Conf. on Smart Grid Comm.}, Nov. 2012, pp. 442--447.

\bibitem{CA2011energy}
``California's energy future: The view to 2050,'' California Council on Science
  and Technology, Tech. Rep., May 2011.

\bibitem{frequencyregulation}
B.~J. Kirby, ``Frequency regulation basics and trends,'' Oak Ridge National
  Laboratory, Tech. Rep. TM2004-291, Dec. 2004.

\bibitem{FERC}
\BIBentryALTinterwordspacing
(2010, May) Federal energy regulatory commission. [Online]. Available:
  \url{http://www.ferc.gov/market-oversight/guide/glossary.asp}
\BIBentrySTDinterwordspacing

\bibitem{PJM}
\BIBentryALTinterwordspacing
PJM. (2015, Apr.) {PJM} manual 11: Energy \& ancillary services market
  operations. [Online]. Available:
  \url{http://www.pjm.com/documents/manuals.aspx}
\BIBentrySTDinterwordspacing

\bibitem{AGC}
A.~R. Bergen and V.~Vittal, \emph{Power Systems Analysis}, 2nd~ed.\hskip 1em
  plus 0.5em minus 0.4em\relax Upper Saddle River, NJ: Prentice Hall, 2000,
  ch.~11.

\bibitem{iagc}
H.~Bevrani and T.~Hiyama, \emph{Intelligent Automatic Generation
  Control}.\hskip 1em plus 0.5em minus 0.4em\relax Boca Raton: CRC Press, 2011.

\bibitem{sellingwind}
E.~Bitar, K.~Poolla, P.~Khargonekar, R.~Rajagopal, P.~Varaiya, and F.~Wu,
  ``Selling random wind,'' in \emph{Proc. 45th Hawaii International Conference
  on Systems Science}, Jan. 2012, pp. 1931--1937.

\bibitem{evnum}
T.~A. Becker and I.~Sidhu, ``Electric vehicles in the united states: A new
  model with forecasts to 2030,'' Center for Entrepreneurship and Technology,
  University of California, Berkeley, Tech. Rep. 2009.1.v.2.0, Aug. 2009.

\bibitem{battery}
C.~Silva, M.~Ross, and T.~Farias, ``Evaluation of energy consumption, emissions
  and cost of plug-in hybrid vehicles,'' \emph{Elsevier Energy Conversion and
  Management}, vol.~50, no.~7, pp. 1635--1643, Jul. 2009.

\bibitem{MISObook}
Y.Chen, M.~Keyser, M.~H. Tackett, R.~Leonard, and J.~Gardner, ``Incorporating
  short-term stored energy resource into {MISO} energy and ancillary service
  market and development of performance-based regulation payment,'' in
  \emph{Energy Storage for Smart Grids}.\hskip 1em plus 0.5em minus 0.4em\relax
  Academic Press, Oct. 2014, ch.~6.

\bibitem{delaware}
\BIBentryALTinterwordspacing
W.~Kempton, V.~Udo, K.~Huber, K.~Komara, S.~Letendre, S.~Baker, D.~Brunner, and
  N.~Pearre. (2009, Jan.) A test of vehicle-to-grid ({V2G}) for energy storage
  and frequency regulation in the {PJM} system. [Online]. Available:
  \url{http://www.udel.edu/V2G/resources/test-v2g-in-pjm-jan09.pdf}
\BIBentrySTDinterwordspacing

\bibitem{financialreturn}
C.~D. White and K.~M. Zhang, ``Using vehicle-to-grid technology for frequency
  regulation and peak-load reduction,'' \emph{J. Power Sources}, vol. 196, pp.
  3972--3980, Apr. 2011.

\bibitem{traveltrends}
A.~Santos, N.~McGuckin, H.~Y. Nakamoto, D.~Gray, and S.~Liss, ``Summary of
  travel trends: 2009 national houseehold travel survey,'' Federal Highway
  Administration, U.S. Department of Transportation, Tech. Rep. FHWA-OL-11-022,
  Jun. 2011.

\bibitem{V2G_fundamentals}
W.~Kempton and J.~Tomic, ``Vehicle-to-grid power fundamentals: Calculating
  capacity and net revenue,'' \emph{J. Power Sources}, vol. 144, no.~1, pp.
  268--279, Jun. 2005.

\bibitem{networkqueueing}
C.-H. Ng and B.-H. Soong, \emph{Queueing Modelling Fundamentals: With
  Applications in Communication Networks}, 2nd~ed.\hskip 1em plus 0.5em minus
  0.4em\relax Hoboken, NJ: Wiley, 2008.

\bibitem{V2G_implementation}
W.~Kempton and J.~Tomic, ``Vehicle-to-grid power implementation: From
  stabilizing the grid to supporting large-scale renewable energy,'' \emph{J.
  Power Sources}, vol. 144, no.~1, pp. 280--294, Jun. 2005.

\bibitem{doublelayer}
A.~Y.~S. Lam, L.~Huang, A.~Silva, and W.~Saad, ``A multi-layer market for
  vehicle-to-grid energy trading in the smart grid,'' in \emph{Proc. 1st IEEE
  INFOCOM Workshop on Green Networking and Smart Grids}, Mar. 2012.

\bibitem{impact_review}
M.~Yilmaz and P.~T. Krein, ``Review of the impact of vehicle-to-grid
  technologies on distribution systems and utility interfaces,'' \emph{IEEE
  Tran. Power Electronics}, vol.~28, no.~12, pp. 5673--5689, Dec. 2013.

\bibitem{EV_loadflow}
R.~Garcia-Valle and J.~G. Vlachogiannis, ``Letter to the editor: Electric
  vehicle demand model for load flow studies,'' \emph{Electric Power Components
  and Systems}, vol.~37, no.~5, pp. 577--582, Sep. 2009.

\bibitem{V2GParking}
U.~C. Chukwu and S.~M. Mahajan, ``{V2G} parking lot with {PV} rooftop for
  capacity enhancement of a distribution system,'' \emph{IEEE Trans.
  Sustainable Energy}, vol.~5, no.~1, pp. 119--127, Jan. 2014.

\bibitem{chargingmodel}
S.~J. Baek, D.~Kim, S.-J. Oh, and J.-A. Jun, ``Modeling of electric vehicle
  charging systems in communications enabled smart grids,'' \emph{IEICE Trans.
  on Information and Systems}, vol. E94-D, no.~8, pp. 1708--1711, Aug. 2011.

\bibitem{optimal_aggregator}
S.~Han, S.~Han, and K.~Sezaki, ``Development of an optimal vehicle-to-grid
  aggregator for frequency regulation,'' \emph{IEEE Trans. Smart Grid}, vol.~1,
  no.~1, pp. 65--72, Jun. 2010.

\bibitem{V2G_QP}
------, ``Optimal control of the plug-in electric vehicles for {V2G} frequency
  regulation using quadratic programming,'' in \emph{Proc. of IEEE PES
  Innovative Smart Grid Technologies}, Jan. 2011.

\bibitem{estimation_capacity}
------, ``Estimation of achievable power capacity from plug-in electric
  vehicles for {V2G} frequency regulation: Case studies for market
  participation,'' \emph{IEEE Trans. Smart Grid}, vol.~2, no.~4, pp. 632--641,
  Dec. 2011.

\bibitem{V2GCap_dynamicEV}
K.~N. Kumar, B.~Sivaneasan, P.~H. Cheah, P.~L. So, and D.~Z.~W. Wang, ``{V2G}
  capacity estimation using dynmaic {EV} scheduling,'' \emph{IEEE Trans. Smart
  Grid}, vol.~5, no.~2, pp. 1051--1060, Mar. 2014.

\bibitem{time-varying_price}
S.~Bashash and H.~K. Fathy, ``Cost-optimal charging of plug-in hybrid electric
  vehicles under time-varying electricity price signals,'' \emph{IEEE Tran.
  Intelligent Transportation Systems}, To appear.

\bibitem{online_scheduling}
J.~Lin, K.-C. Leung, and V.~O.~K. Li, ``Online scheduling for vehicle-to-grid
  regulation service,'' in \emph{Proc. of IEEE Int. Conf. on Smart Grid Comm.},
  Oct. 2013, pp. 43--48.

\bibitem{charging_net}
R.~A. Verzijlbergh, M.~O.~W. Grond, Z.~Lukszo, J.~G. Slootweg, and M.~D. Ilic,
  ``Network impacts and cost savings of controlled ev charging,'' \emph{IEEE
  Trans. Smart Grid}, vol.~3, no.~3, pp. 1203--1212, Sept. 2012.

\bibitem{charging_flow}
L.~Hua, J.~Wang, and C.~Zhou, ``Adaptive electric vehicle charging coordination
  on distribution network,'' \emph{IEEE Trans. Smart Grid}, vol.~5, no.~6, pp.
  2666--2675, Nov. 2014.

\bibitem{charging_price}
Y.~Cao, S.~Tang, C.~Li, P.~Zhang, Y.~Tan, Z.~Zhang, and J.~Li, ``An optimized
  {EV} charging model considering {TOC} price and {SOC} curve,'' \emph{IEEE
  Trans. Smart Grid}, vol.~3, no.~1, pp. 388--393, Mar. 2012.

\bibitem{charging_building}
Y.-M. Wi, J.~uk~Lee, and S.-K. Joo, ``Electric vehicle charging method for
  smart homes/buildings with a photovoltaic system,'' \emph{{IEEE} Trans.
  Consum. Electron.}, vol.~59, no.~2, pp. 323--328, May 2013.

\bibitem{charging_mobility}
M.~Wang, H.~Liang, R.~Zhang, R.~Deng, and X.~Shen, ``Mobility-aware coordinated
  charging for electric vehicles in {VANET}-enhanced smart grid,'' \emph{{IEEE}
  J. Sel. Areas Commun.}, vol.~32, no.~7, pp. 1344--1360, Jul. 2014.

\bibitem{queueingtheory}
D.~Gross, J.~F. Shortle, J.~M. Thompson, and C.~M. Harris, \emph{Fundamentals
  of Queueing Theory}, 4th~ed.\hskip 1em plus 0.5em minus 0.4em\relax Hoboken,
  NJ: Wiley, 2008.

\bibitem{datanetworks}
D.~P. Bertsekas and R.~G. Gallager, \emph{Data Networks}, 2nd~ed.\hskip 1em
  plus 0.5em minus 0.4em\relax Englewood Cliffs, NJ: Prentice Hall, 1992.

\bibitem{SOCdistribution}
R.-C. Leou, C.-L. Su, and C.~N. Lu, ``Stochastic analyses of electric vehicle
  charging impacts on distribution network,'' \emph{IEEE Trans. Power Syst.},
  vol.~29, no.~3, pp. 1055--1063, May 2014.

\bibitem{chargingDuration}
\BIBentryALTinterwordspacing
Wikipedia. (2014, Feb.) Charging station. [Online]. Available:
  \url{http://en.wikipedia.org/wiki/Charging_station}
\BIBentrySTDinterwordspacing

\bibitem{fastcharging}
\BIBentryALTinterwordspacing
{U.S. Department of Energy}. (2014, Feb.) Developing infrastructure to charge
  plug-in electric vehicles. [Online]. Available:
  \url{http://www.afdc.energy.gov/fuels/electricity_infrastructure.html}
\BIBentrySTDinterwordspacing

\bibitem{tesla}
\BIBentryALTinterwordspacing
(2012, May) Tesla motors. [Online]. Available: \url{http://www.teslamotors.com}
\BIBentrySTDinterwordspacing

\bibitem{byd}
\BIBentryALTinterwordspacing
B.~Berman. (2011, 9~Jan.) Exclusive: {BYD} announces breakthrough {U.S.}
  pricing for chinese long-range electric cars. [Online]. Available:
  \url{http://www.plugincars.com/exclusive-byd-executive-provides-breakthrough-us-pricing-chinese-electric-car.html}
\BIBentrySTDinterwordspacing

\end{thebibliography}

\end{document}